\documentclass[graybox]{svmult}

\usepackage{mathptmx}       
\usepackage{helvet}         
\usepackage{courier}        
\usepackage{type1cm}        
\usepackage{makeidx}         
\usepackage{graphicx}        
\usepackage{multicol}        
\usepackage[bottom]{footmisc}
\usepackage{amsmath} 
\usepackage{amssymb} 
\usepackage{bm}
\usepackage{graphicx}
\usepackage{color}
\usepackage{subfig}
\usepackage{url}
\usepackage{floatrow}
\usepackage{tabularx}

\makeindex             

\usepackage[usenames,dvipsnames,svgnames]{xcolor}

\graphicspath{ {images/} }

\newif\ifarxiv

\arxivtrue

\begin{document}

\title*{A BCMP Network Approach\\ to Modeling and Controlling\\ Autonomous Mobility-on-Demand Systems}
\titlerunning{A BCMP Network Approach to Modeling and Controlling AMoD}
\author{Ramon Iglesias, Federico Rossi, Rick Zhang, and Marco Pavone}
\institute{Ramon Iglesias $\cdot$ Federico Rossi $\cdot$  Rick Zhang $\cdot$ Marco Pavone \at Autonomous Systems Laboratory, Stanford University, 496 Lomita Mall, Stanford, CA 94305, \\ \email{{rdit,frossi2,rickz,pavone}@stanford.edu}}
%
%
\maketitle

\abstract*{In this paper we present a queuing network approach to the problem of routing and rebalancing a fleet of self-driving vehicles providing on-demand mobility within a {\em capacitated} road network. We refer to such systems as autonomous mobility-on-demand systems, or AMoD. We first cast an AMoD system into a closed, multi-class BCMP queuing network model. Second, we present analysis tools that allow the characterization of performance metrics for a given routing policy, in terms, e.g., of vehicle availabilities, and first and second order moments of vehicle throughput. Third, we propose a scalable method for the synthesis of routing policies, with performance guarantees in the limit of large fleet sizes. Finally, we validate the theoretical results on a case study of New York City. Collectively, this paper provides a unifying framework for the analysis and control of AMoD systems, which subsumes earlier Jackson and network flow models, provides a quite large set of modeling options (e.g., the inclusion of road capacities and general travel time distributions), and allows the analysis of second and higher-order moments for the performance metrics.}

\abstract{In this paper we present a queuing network approach to the problem of routing and rebalancing a fleet of self-driving vehicles providing on-demand mobility within a {\em capacitated} road network. We refer to such systems as autonomous mobility-on-demand systems, or AMoD. We first cast an AMoD system into a closed, multi-class BCMP queuing network model. Second, we present analysis tools that allow the characterization of performance metrics for a given routing policy, in terms, e.g., of vehicle availabilities, and first and second order moments of vehicle throughput. Third, we propose a scalable method for the synthesis of routing policies, with performance guarantees in the limit of large fleet sizes. Finally, we validate the theoretical results on a case study of New York City. Collectively, this paper provides a unifying framework for the analysis and control of AMoD systems, which subsumes earlier Jackson and network flow models, provides a quite large set of modeling options (e.g., the inclusion of road capacities and general travel time distributions), and allows the analysis of second and higher-order moments for the performance metrics.}

\vspace{-3mm}
\section{Introduction}
\vspace{-3mm}
	\label{sec:intro}
	Personal mobility in the form of privately owned automobiles contributes to increasing levels of traffic congestion, pollution, and under-utilization of vehicles (on average 5\% in the US \cite{DN:15}) -- clearly unsustainable trends for the future. The pressing need to reverse these trends has spurred the creation of cost competitive, on-demand personal mobility solutions such as car-sharing (e.g. Car2Go, ZipCar) and ride-sharing (e.g. Uber, Lyft).
	However, without proper fleet management, car-sharing and, to some extent, ride-sharing systems lead to vehicle imbalances: vehicles aggregate in some areas while becoming depleted in others, due to the asymmetry between trip origins and destinations \cite{RZ-MP:15_MODa}. This issue has been addressed in a variety of ways in the literature. For example, in the context of bike-sharing, \cite{DC-FM-RWC:13} proposes rearranging the stock of bicycles between stations using trucks. The works in \cite{MN-SZ-SB-MJR:15}, \cite{BB-KGZ-NG:15}, and \cite{FA-DDP-AR:14} investigate using paid drivers to move vehicles between car-sharing stations where cars are parked, while \cite{SB-RJ-CR:15} studies the merits of dynamic pricing for incentivizing drivers to move to underserved areas. 
	
	Self-driving vehicles offer the distinctive advantage of being able to rebalance themselves, in addition to the convenience, cost savings, and possibly safety of not requiring a driver. Indeed, it has been shown that one-way vehicle sharing systems with self-driving vehicles (referred to as autonomous mobility-on-demand systems, or AMoD) have the potential to significantly reduce passenger cost-per-mile-traveled, while keeping the advantages and convenience of personal mobility \cite{KS-KT-RZ-EF-DM-MP:14}. Accordingly, a number of works have recently investigated the potential of AMoD systems, with a specific focus on the synthesis and analysis of coordination algorithms. Within this context, the goal of this paper is to provide a principled framework for the analysis and synthesis of routing policies for AMoD systems.

	{\em Literature Review}: Previous work on AMoD systems can be categorized into two main classes: heuristic methods and analytical methods. Heuristic routing strategies are extensively investigated in \cite{DJF-KMK:14,DJF-KMK-PB:15,MWL-TL-SDB-KMK:16} by leveraging a traffic simulator and, in \cite{RZ-FR-MP:16}, by leveraging a model predictive control framework. Analytical models of AMoD systems are proposed in \cite{MP-SLS-EF-DR:12}, \cite{RZ-MP:15_MODa}, and \cite{RZ-FR-MP:16a}, by using fluidic, Jackson queuing network, and capacitated flow frameworks, respectively. Analytical methods have the advantage of providing structural insights (e.g., \cite{RZ-FR-MP:16a}), and provide guidelines for the synthesis of control policies. The problem of controlling AMoD systems is similar to the System Optimal Dynamic Traffic Assignment (SO-DTA) problem (see, e.g., \cite{YC-JB-MM-AP-RB-TW-JH:11,MP:15a}) where the objective is to find optimal routes for all vehicles within congested or capacitated networks such that the total cost is minimized. The main differences between the AMoD control problem and the SO-DTA problem is that SO-DTA only optimizes customer routes, and {\em not} rebalancing routes.
	
	This paper aims at devising a general, unifying analytical framework for analysis and control of AMoD systems, which subsumes many of the analytical models recently presented in the literature, chiefly, \cite{MP-SLS-EF-DR:12}, \cite{RZ-MP:15_MODa}, and \cite{RZ-FR-MP:16a}. Specifically, this paper extends our earlier Jackson network approach in \cite{RZ-MP:15_MODa} by adopting a BCMP queuing-theoretical framework \cite{FB-KMC-RRM-FGP:75,HK-MG:83}. BCMP networks significantly extend Jackson networks by allowing almost arbitrary customer routing and service time distributions, while still admitting a convenient product-form distribution solution for the equilibrium distribution \cite{HK-MG:83}. Such generality allows one to take into account several real-world constraints, in particular road capacities (that is, congestion). Indeed, the impact of AMoD systems on congestion has been a hot topic of debate. For example, \cite{MWL-TL-SDB-KMK:16} notes that empty-traveling rebalancing vehicles may increase congestion and total in-vehicle travel time for customers, but \cite{RZ-FR-MP:16a} shows that, with congestion-aware routing and rebalancing, the increase in congestion can be avoided. The proposed BCMP model recovers the results in \cite{RZ-FR-MP:16a}, with the additional benefits of taking into account the stochasticity of transportation networks and providing estimates for performance metrics.
	
	
	{\em Statement of Contributions}: The contribution of this paper is fourfold. First, we show how an AMoD system can be cast within the framework of closed, multi-class BCMP queuing networks. The framework captures stochastic passenger arrivals, vehicle routing on a road network, and congestion effects. Second, we present analysis tools that allow the characterization of performance metrics for a given routing policy, in terms, e.g., of vehicle availabilities and second-order moments of vehicle throughput. Third, we propose a scalable method for the synthesis of routing policies: namely, we show that, for large fleet sizes, the stochastic optimal routing strategy can be found by solving a linear program. Finally, we validate the theoretical results on a case study of New York City.

	{\em Organization}: The rest of the paper is organized as follows. In Section \ref{background_material}, we cover the basic properties of BCMP networks and, in Section \ref{model_description}, we describe the AMoD model, cast it into a BCMP network, and formally present the routing and rebalancing problem. Section \ref{proposed_solution} presents the mathematical foundations and assumptions required to reach our proposed solution. We validate our approach in Section \ref{numerical_experiments} using a model of Manhattan. Finally, in Section \ref{conclusions}, we state our concluding remarks and discuss potential avenues for future research.

	\vspace{-5mm}
	\section{Background Material}\label{background_material}
	\vspace{-3mm}
	In this section we review some basic definitions and properties of BCMP networks, on which we will rely extensively later in the paper.
	\vspace{-5mm}
	\subsection{Closed, Multi-Class BCMP Networks}
	\vspace{-3mm}
	\label{sec:BCMPnets}
	Let $\mathcal{Z}$ be a network consisting of $N$ independent queues (or nodes). 
	A set of agents move within the network according to a stochastic process, i.e. after receiving service at queue $i$ they proceed to queue $j$ with a given probability. No agent enters or leaves the network from the outside, so the number of agents is fixed and equal to $m$. Such a network is also referred to as a \textit{closed} queuing network. Agents belong to one of $K \in \mathbb{N}_{>0}$ classes, and they can switch between classes upon leaving a node. 
	
	Let $x_{i, k}$ denote the number of agents of class $k \in \{1,\ldots, K\}$ at node $i \in \{1,\ldots, N\}$. The state of node $i$, denoted by $\bm{x}_i$, is given by $\bm{x}_i = (x_{i,1}, ... , x_{i, K}) \in \mathbb{N}^K$. The state space of the network is 
	\cite{EG-GP-JCCN:98}: 
	\[
	\small
	\Omega_m := \{ (\bm{x}_1, ... , \bm{x}_N) : \bm{x}_i \in \mathbb{N}^K, \quad \sum_{i = 1}^N \|\bm{x}_i\|_1 = m \},
	\]
	where $\|\cdot\|_1$ denotes the standard 1-norm (i.e., $\|\bm{x}\|_1 = \sum_{i} |x_i|$). The relative frequency of visits (also known as relative throughput) to node $i$ by agents of class $k$, denoted as $\pi_{i,k}$, is given by the traffic equations \cite{EG-GP-JCCN:98}:
	\begin{equation} \label{eq:throughput}
	\small
	\pi_{i,k} = \sum_{k'= 1}^{K} \sum_{j = 1}^{N} \pi_{j,k'} p_{j,k';i,k}, \quad \text{ for all } i \in \{1,\ldots, N\},
	\end{equation}
	where $p_{j,k'; \, i,k}$ is the probability that upon leaving node $j$, an agent of class $k'$ goes to node $i$ and becomes an agent of class $k$. Equation \eqref{eq:throughput} does not have a unique solution (a typical feature of closed networks), and $\pi = \{\pi_{i,k}\}_{i,k}$ only determines frequencies up to a constant factor (hence the name ``relative" frequency). It is customary to express frequencies in terms of a chosen reference node, e.g., so that $\pi_{1,1} = 1$. 
	
	Queues are allowed to be one of four types: First Come First Serve (FCFS), Processor Sharing, Infinite Server, and Last Arrived, First Served. FCFS nodes have exponentially distributed service times, while the other three queue types may follow any Cox distribution \cite{EG-GP-JCCN:98}. Such a queuing network model is referred to as a closed, multi-class BCMP queuing network \cite{EG-GP-JCCN:98}.

	Let $\mathcal{N}$ represent the set of nodes in the network and $N$ its cardinality. For the remainder of the paper, we will restrict networks to have only two types of nodes: FCFS queues with a single server (for short, SS queues), forming a set $\mathcal S\subset \mathcal N$, and infinite server queues (for short, IS queues), forming a set $\mathcal I\subset \mathcal N$. Furthermore, we consider class-independent and load-independent nodes, whereby at each node $i\in \{1,\ldots, N\}$ the service rate is given by:
	\[
	\small
	\mu_i(x_i) = c_i(x_i)\mu_i^o,
	\]
	where $x_i :=\|\bm{x}_i\|_1$ is the number of agents at node $i$, $\mu_i^o$ is the (class-independent) base service rate, and $c_i(x_i)$ is the (load-independent) {\em capacity} function
	\[
	\small
	c_i(x_i) = \begin{cases}
	x_i & \text{if } x_i \leq c_i^o,\\
	c_i^0 & \text{if } x _i > c_i^o,\\
	\end{cases}
	\]
	which depends on the number of servers $c_i^o$ at the queue. In the case considered in this paper, $c_i^o = 1$ for all $i \in \mathcal{S}$ and $c_i^o = \infty$ for all $i \in \mathcal{I}$.
	
	
	Under the assumption of class-independent service rates, the multi-class network $\mathcal Z$ can be ``compressed" into a single-class network $\mathcal{Z}^*$ with state-space $\Omega_m^* := \{ (x_1, ... , x_N) : x_i \in \mathbb{N}, \quad \sum_{i = 1}^N x_i = m \}$ \cite{KK-MMS:92}. Performance metrics for the original, multi-class network $\mathcal{Z}$ can be found by first analyzing the compressed network $\mathcal{Z}^*$, and then applying suitable scalings for each class. Specifically, let $\small \pi_i = \sum_{k=1}^K \pi_{i,k}$ and $\gamma_i = \sum_{k=1}^K \frac{\pi_{i,k}}{\mu_i^o}$, be the total relative throughput and relative utilization at a node $i$, respectively. Then, the stationary distribution of the compressed, single-class network $\mathcal{Z}^*$ is given by
	\begin{equation*}
	\small
	\mathbb{P}(x_1, ... , x_N) = \frac{1}{G(m)} \prod_{i =1}^{N} \frac{\gamma_{i}^{x_i}}{\prod_{a=1}^{x_i} c_i(a)} , \quad \text{where} \quad G(m) = \sum_{\Omega_m^*} \prod_{i =1}^{N} \frac{\gamma_{i}^{x_i}}{\prod_{a=1}^{x_i} c_i(a)}
	\end{equation*}
	is a normalizing constant.
	Remarkably, the stationary distribution has a product form, a key feature of BCMP networks.
	
	Three performance metrics that are of interest at each node are throughput, expected queue length, and availability. First, the throughput at a node (i.e., the number of agents processed by a node per unit of time) is given by
	\begin{equation}
	\small
	\label{eq:nodethroughput}
	\Lambda_i(m) = \pi_i \, \frac{G(m-1)}{G(m)}.
	\end{equation}
Second, let $\mathbb{P}_i(x_i;\, m)$ be the probability of finding $x_i$ agents at node $i$; then the expected queue length at node $i$ is given by
	$
	\small
	L_i(m) = \sum_{ x_i = 1}^m x_i \mathbb{P}_i(x_i ; \, m).
	$
	
	In the case of IS nodes (i.e., nodes in $\mathcal I$), the expected queue length can be more easily derived via Little's Law as \cite{DKG:12}
	\begin{equation}\label{eq:littles_law}
	\small
		L_i(m) = \Lambda_i(m) / \mu_i^o \quad \text{ for all } i \in \mathcal{I}.
	\end{equation}	
Finally, the availability of single-server, FCFS nodes (i.e., nodes in $\mathcal S$) is defined as the probability that the node has at least 1 agent, and is given by \cite{DKG:12}
	\begin{equation*}
	\small
	A_i(m) = \gamma_i \frac{G(m-1)}{G(m)} \quad \text{ for all } i \in \mathcal{S}.
	\end{equation*}
	The throughputs and the expected queue lengths for the original, multi-class network $\mathcal{Z}^*$ can be found via scaling \cite{KK-MMS:92}, specifically, $\small \Lambda_{i,k}(m) = ({\pi_{i,k}}/{\pi_i}) \Lambda_i(m)$ and $L_{i,k}(m) = ({\pi_{i,k}}/{\pi_i}) L_i(m).$
	
	It is worth noting that evaluating the three performance metrics above requires computation of the normalization constant $G(m)$, which is computationally expensive. However, several techniques are available to avoid the direct computation of $G(m)$. In particular, in this paper we use the Mean Value Analysis method, which, remarkably, can be also used to compute higher moments (e.g., variance) \cite{JS:86}. Details are provided in 
	\ifarxiv
	the Appendix.
\else	
	the extended version of this paper \cite{RI-FR-RZ-MP:16EV}.
\fi
	\vspace{-5mm}
	\subsection{Asymptotic Behavior of Closed BCMP Networks}\label{section:asymp_section}
	\vspace{-5mm}
	
	In this section we describe the asymptotic behavior of closed BCMP networks as the number of agents $m$ goes to infinity. 
	The results described in this section are taken from \cite{DKG:12}, and are detailed for a single-class network; however, as stated in the previous section, results found for a single-class network can easily be ported to the multi-class equivalent in the case of class-independent service rates.
	
	Let $\rho_i := \gamma_i / c_i^o$ be the utilization factor of node $i \in \mathcal N$, where $c_i^o$ is the number of servers at node $i$. Assume that the relative throughputs $\{\pi_i\}_i$ are normalized so that $\max_{ i \in \mathcal S} \, \rho_i = 1$; furthermore, assume that nodes are ordered by their utilization factors so that $1 = \rho_1 \geq \rho_2 \geq \ldots \geq \rho_N$, and define the set of bottleneck nodes as $\mathcal B := \{ i \in \mathcal S : \rho_i =1 \}$.
	
	It can be shown \cite[p. 14]{DKG:12} that, as the number of agents $m$ in the system approaches infinity, the availability at all bottleneck nodes converges to 1 while the availability at non-bottleneck nodes is strictly less than 1, that is
	\begin{equation}\label{eq:availability_lim}
	\small
	\lim_{m \rightarrow \infty} A_i(m) \begin{cases}
	= 1 \quad \forall i \in \mathcal B. \\
	< 1 \quad \forall i \notin \mathcal B.
	\end{cases}
	\end{equation}
	Additionally, the queue lengths at the non-bottleneck nodes have a limiting distribution given by
	\begin{equation}\label{eq:open_network}
	\small
	\lim_{m \rightarrow \infty} \mathbb{P}_i(x_i ; m) = \begin{cases}
	(1 - \rho_i) \, \rho_i^{x_i} & i \in S, i \notin \mathcal{B}, \\
	e^{-\gamma_i}\frac{\gamma_i^{x_i}}{x_i!} & i \in I.
	\end{cases}
	\end{equation}
	Together, \eqref{eq:availability_lim} and \eqref{eq:open_network} have strong implications for the operation of queuing networks with a large number of agents, and in particular for the operation of AMoD systems. Intuitively, \eqref{eq:availability_lim} shows that as we increase the number of agents in the network, they will be increasingly queued at bottleneck nodes, driving availability in those queues to one. Alternatively, non-bottleneck nodes will converge to an availability strictly less than one, implying that there is always a non-zero probability of having an empty queue. In other words, agents will aggregate at the bottlenecks and become depleted elsewhere. Additionally, \eqref{eq:open_network} shows that, as the number of agents goes to infinity, non-bottleneck nodes tend to behave like queues in an equivalent open BCMP network with the bottleneck nodes removed, i.e., individual performance metrics can be calculated in isolation.

\vspace{-5mm}
\section{Model Description and Problem Formulation}\label{model_description}
\vspace{-5mm}
	In this section, we introduce a BCMP network model for AMoD systems, and formalize the problem of routing and rebalancing such systems under stochastic conditions. Casting an AMoD system as a queuing network allows us to characterize and compute key performance metrics including the distribution of the number of vehicles on each road link (a key metric to characterize traffic congestion) and the probability of servicing a passenger request. To emphasize the relationship with the theory presented in the previous section, we reuse the same notation whenever concepts are equivalent.
	
	\vspace{-5mm}
	\subsection{Autonomous Mobility-on-Demand Model}
	\vspace{-5mm}
	\label{amod_model}
	
	Consider a set of stations\footnote{Stations are not necessarily physical locations: they can also be interpreted as a set of geographical regions.} $\mathcal{S}$ distributed within an urban area connected by a network of individual road links $\mathcal{I}$, and $m$ autonomous vehicles providing one-way transportation between these stations for incoming customers. Customers arrive to a station $s\in\mathcal{S}$ with a target destination $t\in\mathcal{S}$ according to a time-invariant Poisson process with rate $\lambda \in \mathbb{R}_{ > 0 }$. The arrival process for all origin-destination pairs is summarized by the set of tuples $\mathcal{Q} = \{(s^{(q)}, t^{(q)}, \lambda^{(q)})\}_q$. 
	
	If on customer arrival there is an available vehicle, the vehicle drives the customer towards its destination. Alternatively, if there are no vehicles, the customer leaves the system (i.e., chooses an alternative transportation system). Thus, we adopt a \emph{passenger loss} model. Such model is appropriate for systems where high quality-of-service is desired; from a technical standpoint, this modeling assumption decouples the passenger queuing process from the vehicle queuing process.
	
	A vehicle driving a passenger through the road network follows a routing policy $\alpha^{(q)}$ (defined in Section \ref{sec:AMoD2BCMP}) from origin to destination, where $q$ indicates the origin-destination-rate tuple. Once it reaches its destination, the vehicle joins the station first-come, first-serve queue and waits for an incoming trip request.
	
	A known problem of such systems is that vehicles will inevitably accumulate at one or more of the stations and reduce the number of vehicles servicing the rest of the system \cite{DKG:12} if no corrective action is taken. To control this problem, we introduce a set of ``virtual rebalancing demands" or ``virtual passengers" whose objective is to balance the system, i.e., to move empty vehicles to stations experiencing higher passenger loss. Similar to passenger demands, rebalancing demands are defined by a set of origin, destination and arrival rate tuples $\mathcal{R} = \{(s^{(r)}, t^{(r)}, \lambda^{(r)})\}_r$, and a corresponding routing policy $\alpha^{(r)}$. Therefore, the objective is to find a set of routing policies $\alpha^{(q)}, \alpha^{(r)}$, for all $q \in \mathcal{Q}$, $r \in \mathcal{R}$, and rebalancing rates $\lambda^{(r)}$, for all $r \in \mathcal{R}$, that balances the system while minimizing the number of vehicles on the road, and thus reducing the impact of the AMoD system on overall traffic congestion.
	\vspace{-5mm}
	\subsection{Casting an AMoD System into a BCMP Network}
	\vspace{-5mm}
	\label{sec:AMoD2BCMP}
	We are now in a position to frame the AMoD system in terms of a BCMP network model. To this end, we represent the vehicles, the road network and the passenger demands in the BCMP framework.
	
	First, the passenger loss assumption allows the model to be characterized as a queuing network with respect only to the \textit{vehicles}. Thus, we will henceforth use the term ``vehicles" to refer to the queuing agents. From this perspective, the stations $\mathcal{S}$ are equivalent to SS queues, and the road links $\mathcal{I}$ are modeled as IS queues.
	
	Second, we map the underlying road network to a directed graph with the queues as \emph{edges}, and introduce the set of road intersections $\mathcal{V}$ to function as graph vertices. As in Section \ref{background_material}, the set of all queues is given by $\mathcal{N}= \{\mathcal{S} \cup \mathcal{I}\}$. Let Parent($i$) and Child($i$) be the origin and destination vertices of edge $i$. Then, a road that goes from intersection $j$ to intersection $l$ is represented by a queue $i\in\mathcal{I}$ such that Parent($i$)$= j$ and Child($i$)$= l$. Note that the road may not have lanes in the opposite direction, in which case a queue $i'$ with Parent($i'$)$= l$ and Child($i'$)$=j$ would not exist. For example, in Figure \ref{fig:network}, queue 14 starts at vertex 1 and ends at vertex 5. However, there is no queue that connect the vertices in the opposite direction. Similarly, we assume that stations are adjacent to road intersections, and therefore stations are modeled as edges with the same parent and child vertex. An intersection may have access to either one station (e.g., vertex 2 in Figure \ref{fig:network}), or zero stations (e.g., vertex 5 in Figure \ref{fig:network}).
	\begin{figure}[h]
		\sidecaption[t]
		\includegraphics[height=0.22\textheight]{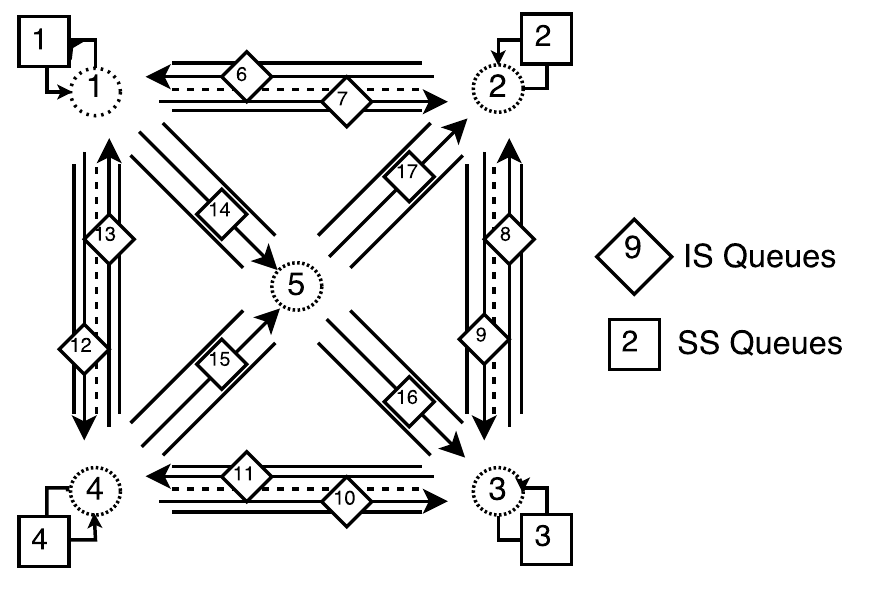}
		\caption{BCMP network model of an AMoD system. Diamonds represent infinite-server road links, squares represent the single-server vehicle stations, and dotted circles represent road intersections.}
		\label{fig:network}
	\end{figure}
	
	Third, we introduce classes to represent the process of choosing destinations. We map the set of tuples $\mathcal{Q}$ and $\mathcal{R}$ defined in Section \ref{amod_model} to a set of classes $\mathcal{K}$ such that $\mathcal{K} = \{\mathcal{Q} \cup \mathcal{R}\}$. Moreover, let $\mathcal{O}_i$ be the subset of classes whose origin $s^{(k)}$ is the station $i$, such that $\mathcal{O}_i := \{k \in \mathcal{K} : s^{(k)}=i\}$ and $\mathcal{D}_i$ be the subset whose destination $t^{(k)}$ is the station $i$, such that $\mathcal{D}_i := \{k \in \mathcal{K} : t^{(k)}=i\}$.
	Thus, the probability that a vehicle at station $i$ will leave for station $j$ with a (real or virtual) passenger is the ratio between the respective (real or virtual) arrival rate $\lambda^{(k)}$, with $s^{(k)}=i$, $t^{(k)}=j$, and the sum of all arrival rates at station $i$. Formally, the probability that a vehicle of class $k$ switches to class $k'$ upon arrival to its destination $t^{(k)}$ is
	$
	\widetilde{p}_{t^{(k)}}^{(k')} = \left(\lambda^{(k')}/\widetilde{\lambda}_{t^{(k)}}\right),
	$
	where $\widetilde{\lambda}_i = \sum_{k \in \mathcal{O}_i} \lambda^{(k)}$ is the sum of all arrival rates at station $t^{(k)}$. Consequently, at any instant in time a vehicle belongs to a class $k \in \mathcal{K}$, regardless of whether it is waiting at a station or traveling along the road network. By switching class on vehicle arrival, the vehicles's transition probabilities $\widetilde{p}_{t^{(k)}}^{(k')}$ encode the passenger and rebalancing demands defined in Section \ref{amod_model}.	
	
	As mentioned in the previous section, the traversal of a vehicle from its source $s^{(k)}$ to its destination $t^{(k)}$ is guided by a routing policy $\alpha^{(k)}$. This routing policy, in queuing terms, consists of a matrix of transition probabilities. Let $\mathcal{W}_i = \{j \in \mathcal{N}: \text{Parent($j$)} = i\}$ be the set of queues that begin in vertex $i$, and 	$\mathcal{U}_i = \{j \in \mathcal{N}: \text{Child($j$)} = i\}$ the set of queues that end in vertex $i$. A vehicle of class $k$ leaves the station $s^{(k)}$ via one of the adjacent roads $j \in \mathcal{W}_{\text{Child}(s^{(k)})}$ with probability $\alpha_{s^{(k)},j}^{(k)}$. It continues traversing the road network via these adjacency relationships following the routing probabilities $\alpha_{i,j}^{(k)}$ until it is adjacent to its goal $t^{(k)}$. At this point, the vehicle proceeds to the destination and changes its class to $k' \in \mathcal{O}_{t^{(k)}}$ with probability $\widetilde{p}_{t^{(k)}}^{(k')}$. This behavior is encapsulated by the routing matrix
	\begin{equation}\label{eq:routing} 
	\small
	p_{i,k;j,k'} = \begin{cases}
	\alpha_{i,j}^{(k)} &\text{if $k = k'$, $j \in \mathcal{W}_{\text{Child}(i)}$, $t^{(k)} \notin \mathcal{W}_{\text{Child}(i)}$},\\
	\widetilde{p}_{j}^{(k')} &\text{if $j = t^{(k)}$, $t^{(k)} \in \mathcal{W}_{\text{Child}(i)}$, $k' \in \mathcal{O}_j$}, \\
	0 &\text{otherwise},
	\end{cases}
	\end{equation}
	such that $\sum_{j,k'} p_{i,k;j,k'} = 1$. Thus, the relative throughput $\pi_{i,k}$, total relative throughput $\pi_i$, and utilization $\gamma_i$ have the same definition as in Section \ref{background_material}.
	
	As stated before, the queuing process at the stations is modeled as a SS queue where the service rate of the vehicles $\mu_i(a)$ is equal to the sum of real and virtual passenger arrival rates, i.e. $\mu_i(a) = \widetilde{\lambda}_i$ for any station $i$ and queue length $a$. Additionally, by modeling road links as IS queues, we assume that their service rates follow a Cox distribution with mean $\small \mu_{i}(a) = \frac{c_{i}(a)}{T_{i}}$, where $T_{i}$ is the expected time required to cross link $i$ in absence of congestion, and $c_{i}(a)$ is the capacity factor when there are $a$ vehicles in the queue. In this paper, we only consider the case of load-independent travel times, therefore $c_{i}(a) = a$ for all $ a$, i.e., the service rate is the same regardless of the number of vehicles on the road. We do not make further assumptions on the distribution of the service times. The assumption of load-independent travel times is representative of uncongested traffic \cite{BPR:64}: in Section \ref{section:problem_formulation} we discuss how to incorporate probabilistic constraints for congestion on road links.
	
	\vspace{-5mm}
	\subsection{Problem formulation}\label{section:problem_formulation}
	\vspace{-5mm}
	
	As stated in Equation (\ref{eq:availability_lim}), vehicles tend to accumulate in bottleneck stations driving their availability to 1 as the fleet size increases, while the rest of the stations have availability strictly smaller than 1. In other words, for unbalanced systems, availability at most stations is capped regardless of fleet size. Therefore, it is desirable to make all stations "bottleneck" stations, i.e., set the constraint $\gamma_i = \gamma_j$ for all $i,j \in \mathcal{S}$, so as to (i) enforce a natural notion of ``service fairness,'' and (ii) prevent needless accumulation of empty vehicles at the stations.
	
	However, it is desirable to minimize the impact that the rebalancing vehicles have on the road network.
	We achieve this by minimizing the expected number of vehicles on the road serving customer and rebalancing demands. Using Equation (\ref{eq:littles_law}), the expected amount of vehicles on a given road link $i$ is given by $\Lambda_i(m)T_i$. 
	
	Lastly, we wish to avoid congestion on the individual road links. Traditionally, the relation between vehicle flow and congestion is parametrized by two basic quantities: the \textit{free-flow travel time} $T_i$, i.e., the time it takes to traverse a link in absence of other traffic; and the \textit{nominal capacity} $C_i$, i.e., the measure of traffic flow beyond which travel time increases very rapidly \cite{MP:15a}. Assuming that travel time remains approximately constant when traffic is below the nominal capacity (an assumption typical of many state-of-the-art traffic models \cite{MP:15a}), our approach is to keep the expected traffic $\Lambda_i(m)T_i$ below the nominal capacity $C_i$ and thus avoid congestion effects. Note that by constraining in expectation there is a non-zero probability of exceeding the constraint; however, in Section \ref{section:asymp_opt_sol}, we show that, asymptotically, it is also possible to constrain the \emph{probability} of exceeding the congestion constraint.
	
	Accordingly, the routing problem we wish to study in this paper (henceforth referred to as the \emph{Optimal Stochastic Capacitated AMoD Routing and Rebalancing problem}, or OSCARR) can now be formulated as follows: 
		\begin{subequations}\label{eq:cbrrp}
		{\small
		\begin{align}
		& \underset{\lambda^{(r \in \mathcal{R})}, \alpha_{ij}^{(k \in \mathcal{K})}}{\text{minimize}}
		& & \sum_{i \in \mathcal{I}}\Lambda_i(m)T_i, \nonumber\\
		& \text{subject to}
		& & \gamma_i = \gamma_j,& i,j \in \mathcal{S}, \label{eq:equal_utilization}\\
		&
		& & \Lambda_i(m)T_i \leq C_i, & i \in \mathcal{I}, \label{eq:congestion_constraint}\\
		&
		& & \pi_{s^{(k)}, k} = \sum_{k' \in \mathcal{K}} \sum_{j \in \mathcal{N}} \pi_{j,k} p_{j,k;t^{(k)},k'}, & k \in \mathcal{K}, \label{eq:routing_constraint}\\
		&
		& & \pi_{i,k} = \sum_{k'= 1}^{K} \sum_{j = 1}^{N} \pi_{j,k'} p_{j,k';i,k} & i \in \{\mathcal{S} \cup \mathcal{I} \}, \label{eq:traffic_eqs} \\
		&
		& & \sum_{j} \alpha_{ij}^{(k)} = 1, \quad \alpha_{ij}^{(k)} \ge 0, & i,j \in \{\mathcal{S} \cup \mathcal{I} \}, \label{eq:probabilistic_routing}\\
		& 
		& & \lambda_r \ge 0, & r \in \mathcal{R}. \label{eq:nonnegative_demands}
		\end{align}}
		\end{subequations}
		Constraint (\ref{eq:equal_utilization}) enforces equal availability at all stations, while constraint (\ref{eq:congestion_constraint}) ensures that all road links are (on average) uncongested.
	Constraints (\ref{eq:routing_constraint})--(\ref{eq:nonnegative_demands}) enforce consistency in the model. Namely, (\ref{eq:routing_constraint}) ensures that all traffic leaving the source $s^{(k)}$ of class $k$ arrives at its destination $t^{(k)}$, (\ref{eq:traffic_eqs}) enforces the traffic equations \eqref{eq:throughput}, (\ref{eq:probabilistic_routing}) ensures that $\alpha^{(k)}_{ij}$ is a valid probability measure, and (\ref{eq:nonnegative_demands}) guarantees nonnegative rebalancing rates.
	
	At this point, we would like to reiterate some assumptions built into the model. 
	First, the proposed model is time-invariant. That is, we assume that customer and rebalancing rates remain constant for the segment of time under analysis, and that the network is able to reach its equilibrium distribution. An option for including the variation of customer demand over time is to discretize a period of time into smaller segments, each with its own arrival parameters and resulting rebalancing rates. 
	Second, the passenger loss model assumes impatient customers and is well suited for cases where high level of service is required. This allows us to simplify the model by focusing only on the vehicle process; however, it disregards the fact that customers may have different waiting thresholds and, consequently, the queuing process of waiting customers.
	Third, we focus on keeping traffic within the nominal road capacities in expectation, allowing us to assume load-independent travel times and to model exogenous traffic as a reduction in road capacity.
	Finally, we make no assumptions on the distribution of travel times on the road links: the analysis proposed in this paper captures arbitrary distributions of travel times and only depends on the \emph{mean} travel time. 

	\vspace{-5mm}
	\section{Asymptotically Optimal Algorithms for AMoD routing}\label{proposed_solution}
	\vspace{-5mm}
	
	
	In this section we show that, as the fleet size goes to infinity, the solution to OSCARR can be found by solving a linear program. This insight allows the efficient computation of asymptotically optimal routing and rebalancing policies and of the resulting performance parameters for AMoD systems with very large numbers of customers, vehicles and road links. 
	
	First, we introduce simplifications possible due to the nature of the routing matrix $\{\alpha_{i,j}^{(k)}\}_{(i,j),k}$. Then, we express the problem from a flow conservation perspective. Finally, we show that the problem allows an asymptotically optimal solution with bounds on the probability of exceeding road capacities. The solution we find is equivalent to the one presented in \cite{RZ-FR-MP:16a}: thus, we show that the network flow model in \cite{RZ-FR-MP:16a} also captures the asymptotic behavior of a stochastic AMoD routing and rebalancing problem.

	\vspace{-5mm}	
	\subsection{Folding of traffic equations}
	\vspace{-5mm}
	The next two lemmas show that the traffic equations (\ref{eq:throughput}) at the SS queues can be expressed in terms of other SS queues, and that the balanced network constraint can be expressed in terms of real and virtual passenger arrivals. The proof of Lemmas \ref{le:folding} and \ref{le:lambdas_constr} are omitted for space reasons and can be found in \ifarxiv
	the Appendix.\else the extended version of this paper \cite{RI-FR-RZ-MP:16EV}.\fi
	\begin{lemma}[Folding of traffic equations]
	\label{le:folding}
		Let $\mathcal{Z}$ be a feasible solution to OSCARR. Then, the relative throughputs of the single server stations can be expressed in terms of the relative throughputs of the other single server stations, that is
		\begin{equation}
		\small
		\label{eq:folded_traffic_origins_lemma}
		\pi_i = \sum_{k \in \mathcal{D}_i} \widetilde{p}_{s^{(k)}}^{(k)} \pi_{s^{(k)}}, \quad \text{for all } i \in S.
		\end{equation}
	\end{lemma}
	\begin{lemma}[Balanced system in terms of arrival rates]
	\label{le:lambdas_constr} 
		Let $\mathcal{Z}$ be a feasible solution to OSCARR, then the constraint $\gamma_i = \gamma_j\,\,\forall \, i,j,$ is equivalent to
		\begin{equation}\label{eq:lambdas_constr} 
		\small
			 \widetilde{\lambda}_i = \sum_{k \in \mathcal{D}_i} \lambda_{s^{(k)}}^{(k)}.
		\end{equation}
	\end{lemma}
	
	\vspace{-5mm}	
	\subsection{Asymptotically Optimal Solution}
	\vspace{-5mm}
	\label{section:asymp_opt_sol}
	
	As discussed in Section \ref{sec:BCMPnets}, relative throughputs are computed up to a constant multiplicative factor.	
	Thus, without loss of generality, we can set the additional constraint $\pi_{s^{(1)}} = \widetilde{\lambda}_1$, which, along with (\ref{eq:equal_utilization}), implies that 	
	\begin{equation}\label{eq:unit_util}
	\small
		\pi_i = \widetilde{\lambda}_i, \quad \pi_{s^{(k)},k} = \lambda^{(k)}, \quad \text{and} \quad \gamma_i = 1, \quad \text{for all } i \in \mathcal{S}.
	\end{equation}
	As seen in Section \ref{section:asymp_section}, the availabilities of stations with the highest relative utilization tend to one as the fleet size goes to infinity. Since the stations are modeled as single-server queues, $\rho_i = \gamma_i$ for all $i \in \mathcal{S}$. Therefore, if the system is balanced, $\gamma_i = \gamma_S^\text{max} = \gamma = 1$ for all $i \in \mathcal{S}$. That is, the set of bottleneck stations $\mathcal{B}$ includes all stations in $\mathcal{S}$ and $\lim_{m \to \infty} \frac{G(m-1)}{G(m)} = 1$ by Equation \eqref{eq:availability_lim} . 
	
	As $m \rightarrow \infty$ and $\frac{G(m-1)}{G(m)} \rightarrow 1$, the throughput at every station $\Lambda_i(m)$ becomes a linear function of the relative frequency of visits to that station, according to Equation \eqref{eq:nodethroughput}. Thus, the objective function and the constraints in \eqref{eq:cbrrp} are reduced to linear functions. We define the resulting problem (i.e., Problem \eqref{eq:cbrrp} with $G(m-1)/G(m)=1$) as the \emph{Asymptotically Optimal Stochastic Capacitated AMoD Routing and Rebalancing problem}, or A-OSCARR. The following lemma shows that the optimal solution to OSCARR approaches the optimal solution to A-OSCARR as $m$ increases.

	\begin{lemma}[Asymptotic behavior of OSCARR]
		Let $\{\pi_{i,k}^{*}(m)\}_{i,k}$ be the set of relative throughputs corresponding to an optimal solution to OSCARR with a given set of customer demands $\{\lambda_i\}_{i}$ and a fleet size $m$. Also, let $\{\hat{\pi}_{i,k}\}_{i,k}$ be the set of relative throughputs corresponding to an optimal solution to A-OSCARR for the same set of customer demands. Then,
		\begin{equation}\label{eq:asymptotical_optimum}
		\small
		\begin{split}
		\lim_{m \to \infty} & \frac{G(m-1)}{G(m)}\sum_{i \in I} T_{i} \sum_{k \in \mathcal{K}} \pi_{i,k}^{*} = \sum_{i \in I} T_{i} \sum_{k \in \mathcal{K}} \hat{\pi}_{i,k}\,.
		\end{split}
		\end{equation}
	\end{lemma}
	\begin{proof}
		We arrive to the proof by contradiction. Recall that $\pi_i = \sum_{k \in \mathcal{K}} \pi_{i,k}$. Assume Equation (\ref{eq:asymptotical_optimum}) did not hold. By definition, the following equations hold for all $m$ and $\{\pi_{i,k}\}_{i,k}$:

		\[
\setlength{\abovedisplayshortskip}{\abovedisplayskip}
\setlength{\belowdisplayshortskip}{\belowdisplayskip}
\begin{minipage}{.55\textwidth}
  \vspace*{-\baselineskip}
  \begin{equation}\label{eq:def_fsharp}
    \frac{G(m-1)}{G(m)}\sum_{i \in I} T_{i} \pi_{i}^{*} \leq \frac{G(m-1)}{G(m)}\sum_{i \in I} T_{i} \pi_{i}\,,
  \end{equation}
\end{minipage}
%
%
\begin{minipage}{.45\textwidth}
  \vspace*{-\baselineskip}
  \begin{equation}\label{eq:def_fa}
  \sum_{i \in I} T_{i} \hat{\pi}_{i} \leq \sum_{i \in I} T_{i} \pi_{i}\,.
  \end{equation}
\end{minipage}
\]


		Applying the limit to (\ref{eq:def_fsharp}) and using (\ref{eq:availability_lim}), we obtain $\small \sum_{i \in I} T_{i} \lim_{m \to \infty}(\pi_{i}^{*}) \leq \sum_{i \in I} T_{i} \pi_{i}.$
		However, according to our assumption, either $\small \sum_{i \in I} T_{i} \lim_{m \to \infty}(\pi_{i}^{*}) > \sum_{i \in I} T_{i} \hat{\pi}_{i}$ or $ \small \sum_{i \in I} T_{i} \lim_{m \to \infty}(\pi_{i}^{*}) < \sum_{i \in I} T_{i} \hat{\pi}_{i}$ but the former violates Equation \eqref{eq:def_fsharp}, and the latter \eqref{eq:def_fa}.
	\end{proof}

	As discussed in Section \ref{section:problem_formulation}, constraint \ref{eq:congestion_constraint} only enforces an upper bound on the expected number of vehicles traversing a link. However, in the asymptotic regime, it is possible to enforce an analytical upper bound on the \emph{probability} of exceeding the nominal capacity of any given road link. As seen in Equation (\ref{eq:open_network}), as the fleet size increases, the distribution of the number of vehicles on a road link $i$ converges to a Poisson distribution with mean $T_i \pi_i $. The cumulative density function of a Poisson distribution is given by $Pr(X < \bar{x}) = Q(\lfloor\bar{x} + 1\rfloor, \tilde C)$, where $\tilde C$ is the mean of the distribution and $Q$ is the regularized upper incomplete gamma function. Let $\epsilon$ be the maximum tolerable probability of exceeding the nominal capacity. Set $\widehat{C}_{i} = Q^{-1}( 1 - \epsilon; \lfloor C_{i} + 1\rfloor)$, i.e. $Q(\lfloor C_{i} + 1\rfloor,\widehat{C}_{i})=1-\epsilon$. Then the constraint $\Lambda_i(m)T_i \leq \widehat{C}_{i}$ is equivalent to $\lim_{m \rightarrow \infty} \mathbb{P}_i(x_i < C_i ; m) \geq 1 - \epsilon$.

	\vspace{-5mm}
	\subsection{Linear programming formulation and multi-commodity flow equivalence}\label{section:flow_cons}
	\vspace{-5mm}
	We now show that an asymptotically optimal routing and rebalancing problem can be framed as a multi-commodity flow problem. Specifically, we show that A-OSCARR is equivalent to the Congestion-Free Routing and Rebalancing problem presented in \cite{RZ-FR-MP:16a}: thus, (i) A-OSCARR can be solved efficiently by ad-hoc algorithms for multi-commodity flow (e.g. \cite{AVG-JDO-SP-CS:98}) and (ii) the theoretical results presented in \cite{RZ-FR-MP:16a} (namely, the finding that rebalancing trips do not increase congestion) extend, in expectation, to \emph{stochastic} systems.
	
	First, we show that the problem can be solved exclusively for the relative throughputs on the road links, and then we show that the resulting equations are equivalent to a minimum cost, multi-commodity flow problem.

	The relative throughput going from an intersection $i$ into adjacent roads is $\sum_{j \in \mathcal{W}_{i}'} \pi_{j,k}$, where $\mathcal{W}_i' = \{\mathcal{W}_i \cap \mathcal{I}\}$ is the set of road links that begin in node $i$. Similarly, the relative throughput entering the intersection $i$ from the road network is $\sum_{j \in \mathcal{U}_{i}'} \pi_{j,k}$, where $\mathcal{U}_i' = \{\mathcal{U}_i \cap \mathcal{I}\}$ is the set of road links terminating in $i$. Additionally, define $d_i^{(k)}$ as the difference between the relative throughput leaving the intersection and the relative throughput entering the intersection. From (\ref{eq:traffic_eqs}), (\ref{eq:routing_constraint}), and (\ref{eq:unit_util}), we see that for customer classes
	\begin{equation*}
	\small
	\sum_{j \in \mathcal{W}_{i}'} \pi_{j,q} - \sum_{j \in \mathcal{U}_{i}'} \pi_{j,q} = d_i^{(q)}, \quad \text{where} \quad
	d_i^{(q)} = \begin{cases}
	\lambda^{(q)} &\text{if $i = s^{(q)}$},\\
	- \lambda^{(q)} &\text{if $i = t^{(q)}$}, \\
	0 &\text{otherwise}.
	\end{cases}
	\end{equation*}
	While the rebalancing arrival rates $\lambda^{(r)}$ are not fixed, we do know from Equation \eqref{eq:routing_constraint} and from the definition of $d_i^{(q)}$ that $d_{s^{(r)}}^{(r)} = -d_{t^{(r)}}^{(r)}$. Thus,
	\begin{equation*}
	\small
	\sum_{j \in \mathcal{W}_{s^{(r)}}'} \pi_{j,r} - \sum_{j \in \mathcal{U}_{s^{(r)}}'} \pi_{j,r} = - \sum_{j \in \mathcal{W}_{t^{(r)}}'} \pi_{j,r} + \sum_{j \in \mathcal{U}_{t^{(r)}}'} \pi_{j,r}. 
	\end{equation*}
	Finally, we can rewrite Lemma \ref{le:lambdas_constr} as
		\[
		\small
		\sum_{q \in \mathcal{Q}}d_i^{(q)} + \sum_{r \in \mathcal{R}}\sum_{j \in \mathcal{W}_{i}'} \pi_{j,r} - \sum_{j \in \mathcal{U}_{i}'} \pi_{j,r} = 0.
		\]
	Thus, in the asymptotic regime Problem (\ref{eq:cbrrp}) can be restated as
	\begin{subequations}\label{eq:cbrrp_matrix}
	\small
	\begin{align}
	& \underset{\pi_{i \in \mathcal{I},k \in \mathcal{K}}}{\text{minimize}}
	& & \!\!\!\sum_{i \in I} T_{i} \sum_{k \in \mathcal{K}} \pi_{i,k}, \nonumber\\
	& \text{subject to}
	& & \!\!\!\sum_{q \in \mathcal{Q}}d_i^{(q)} + \sum_{r \in \mathcal{R}}\sum_{j \in \mathcal{W}_{i}'} \pi_{j,r} - \sum_{j \in \mathcal{U}_{i}'} \pi_{j,r} = 0 & & \forall i \in \mathcal{S}, \label{eq:equal_utilization_flow}\\	
	&
	& & \!\!\!T_{i} \sum_{k \in \mathcal{K}} \pi_{j,k} \leq \widehat{C}_{i} & & \forall i \in \mathcal{I}, \label{eq:congestion_constraint_flow}\\
	&
	& & \!\!\!\sum_{j \in \mathcal{W}_{i}'} \pi_{j,q} - \sum_{j \in \mathcal{U}_{i}'}\pi_{j,q} = d_i^{(q)} & &\forall i \in \mathcal{S}, \label{eq:customer_class_flow}\\
	&
	& & \!\!\!\sum_{j \in \mathcal{W}_{s^{(r)}}'} \!\!\! \pi_{j,r} -\!\!\! \sum_{j \in \mathcal{U}_{s^{(r)}}'} \!\!\! \pi_{j,r} = \!\!\!\sum_{j \in \mathcal{U}_{t^{(r)}}'} \!\!\! \pi_{j,r} - \!\!\!\sum_{j \in \mathcal{W}_{t^{(r)}}'} \!\!\! \pi_{j,r} & & \forall r \in \mathcal{R}, \label{eq:routing_constraint_flow}\\
	&
	& & \!\!\!\sum_{j \in \mathcal{W}_{i}'} \pi_{j,r} - \sum_{j \in \mathcal{U}_{i}'}\pi_{j,r} = 0 & &\forall i \in \mathcal{S} \setminus \{s^{(r)},t^{(r)}\}, \label{eq:traffic_eqs_flow} \\
	&
	& & \!\!\!\sum_{j \in \mathcal{W}_{s^{(r)}}'} \pi_{j,r} - \sum_{j \in \mathcal{U}_{s^{(r)}}'} \pi_{j,r} \geq 0 && \forall r \in \mathcal{R}, \label{eq:nonnegative_demands_flow}\\
	&
	& & \!\!\!\pi_{i,k} \ge 0, & &\forall i \in \mathcal{I}, k \in \mathcal{K}.
	\end{align}
	\end{subequations}
	Here, constraints (\ref{eq:equal_utilization_flow}) and (\ref{eq:congestion_constraint_flow}) are direct equivalents to (\ref{eq:equal_utilization}) and (\ref{eq:congestion_constraint}), respectively. By keeping traffic continuity and equating throughputs at source and target stations, (\ref{eq:customer_class_flow}) enforces (\ref{eq:routing_constraint}) and (\ref{eq:traffic_eqs}) for the customer classes. For the rebalancing classes, (\ref{eq:routing_constraint_flow}) is equivalent to (\ref{eq:routing_constraint}) and (\ref{eq:traffic_eqs_flow}) to (\ref{eq:traffic_eqs}). Non-negativity of rebalancing rates (\ref{eq:nonnegative_demands}) is kept by (\ref{eq:nonnegative_demands_flow}).
	
	Thus, A-OSCARR can be solved efficiently as a linear program.
	Note that this formulation is very similar to the multi-commodity flow found in \cite{RZ-FR-MP:16a}. The formulation in this paper prescribes specific routing policies for distinct rebalancing origin-destination pairs, while \cite{RZ-FR-MP:16a} only computes a single ``rebalancing flow": however, stochastic routing policies can be computed from the rebalancing flow in \cite{RZ-FR-MP:16a} with a flow decomposition algorithm \cite{LRF-DRF:62}.

	\vspace{-5mm}
	\section{Numerical Experiments}\label{numerical_experiments}
	\vspace{-5mm}

	To illustrate a real-life application of the results in this paper, we applied our model to a case study of Manhattan, and computed the system performance metrics as a function of fleet size using the Mean Value Analysis. Results show that the solution correctly balances vehicle availability across stations while keeping road traffic within the capacity constraints, and that the assumption of load-independent travel times is relatively well founded.
	\vspace{-2mm}
	\begin{figure}[H]
		\centering
		\ifarxiv \includegraphics[width=0.9\textwidth,trim={0 0 0 0},clip]{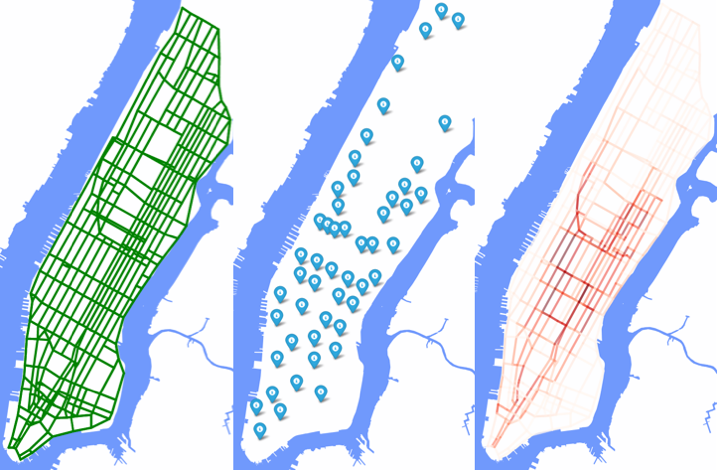}
		\else
		\vspace{-2mm}
		\includegraphics[width=0.83\textwidth,trim={0 5mm 0 75mm},clip]{manhattan_pics_min.png}
		\fi
		\caption{Manhattan scenario \ifarxiv\else  (detail)\fi. Left: modeled road network. Center: Station locations. Right: Resulting vehicular flow (darker flows show higher vehicular presence).}
		\label{fig:manhattan_pics}
	\end{figure}
	The model used for this case study consists of a simplified version of Manhattan's real road network, with 1005 road links and 357 nodes. To select station positions and compute the rates $\lambda^{(q)}$ of the origin-destination flows $\mathcal{Q}$, we used the taxi trips within Manhattan that took place between 7:00AM and 8:00AM on March 1, 2012 (22,416 trips) from the New York City Taxi and Limousine Commission dataset\footnote{http://www.nyc.gov/html/tlc/html/about/trip\_record\_data.shtml}. We clustered the pickup locations into 50 different groups with K-means clustering, and placed a station at the road intersection closest to each cluster centroid. We fitted an origin-destination model with exponential distributions to describe the customer trip demands between the stations.
	Road capacities were reduced
	to ensure that the model reaches maximum utilization in some road links; in the real world, a qualitatively similar reduction in road capacity is caused by traffic exogenous to the taxi system.

	We considered two scenarios: the ``baseline'' scenario where traffic constraints on each road link are based on expectation, i.e., on average the number of vehicles on a road link is below its nominal capacity; and the ``conservative'' scenario where the constraints are based on the asymptotic probability of exceeding the nominal capacity (specifically, the asymptotic probability of exceeding the nominal capacity is constrained to be lower than 10\%). Figure \ref{fig:manhattan_pics} shows the station locations, the road network, and the resulting traffic flow, and Figure \ref{fig:sub1} shows our results. 

	We see from Figure \ref{fig:manhattan_av} that, as intended, the station availabilities are balanced and approach one as the fleet size increases. However, Figure \ref{fig:manhattan_veh} shows that there is a trade off between availability and vehicle utilization. For example, for a fleet size of 4,000 vehicles, half of the vehicles are expected to be waiting at the stations. In contrast, a fleet of 2,400 vehicles results in availability of 91\% and only 516 vehicles are expected to be at the stations. Not shown in the figures, 34\% of the trips are for rebalancing purposes; in contrast, only about 18\% of the traveling vehicles are rebalancing. This shows that rebalancing trips are significantly shorter than passenger trips, in line with the goal of minimizing the number of empty vehicles on the road and thus road congestion.

	Although Figures \ref{fig:manhattan_av} and \ref{fig:manhattan_veh} show only the results for the baseline case, for the conservative scenario the difference in availabilities is less than 0.1\%, and the difference in the total number of vehicles on the road is less than 7, regardless of fleet size. However, road utilization is significantly different in the two scenarios we considered. In Figure \ref{fig:manhattan_ut}, we see that, as the fleet size increases, the likelihood of exceeding the nominal capacity
	approaches 50\%. In contrast, in the conservative scenario, the probability of exceeding the capacity is never more than 10\% --by design-- regardless of fleet size. 
	
	Lastly, we evaluated how much the assumption of load-independent travel times deviates from the more realistic case where travel time depends on traffic. Assuming asymptotic conditions (i.e., the number of vehicles on each road follows a Poisson distribution), we computed for both scenarios the expected travel time between each origin-destination pair by using the Bureau of Public Roads (BPR) delay model \ifarxiv(described in 
	the Appendix)\else \cite{BPR:64}\fi, and estimated the difference with respect to the load-independent travel time used in this paper. The results, depicted in Figure \ref{fig:manhattan_traveltimes}, show that the maximum difference for the baseline and conservative scenarios are an increase of around 8\% and 4\%, respectively, and the difference tends to be smaller for higher trip times. Thus, for this specific case study, our assumption is relatively well founded.
	\vspace{-5.5mm}
	\begin{figure}[H]
		\centering
		\ifarxiv
		\sidesubfloat[]{\includegraphics[width=.45\textwidth]{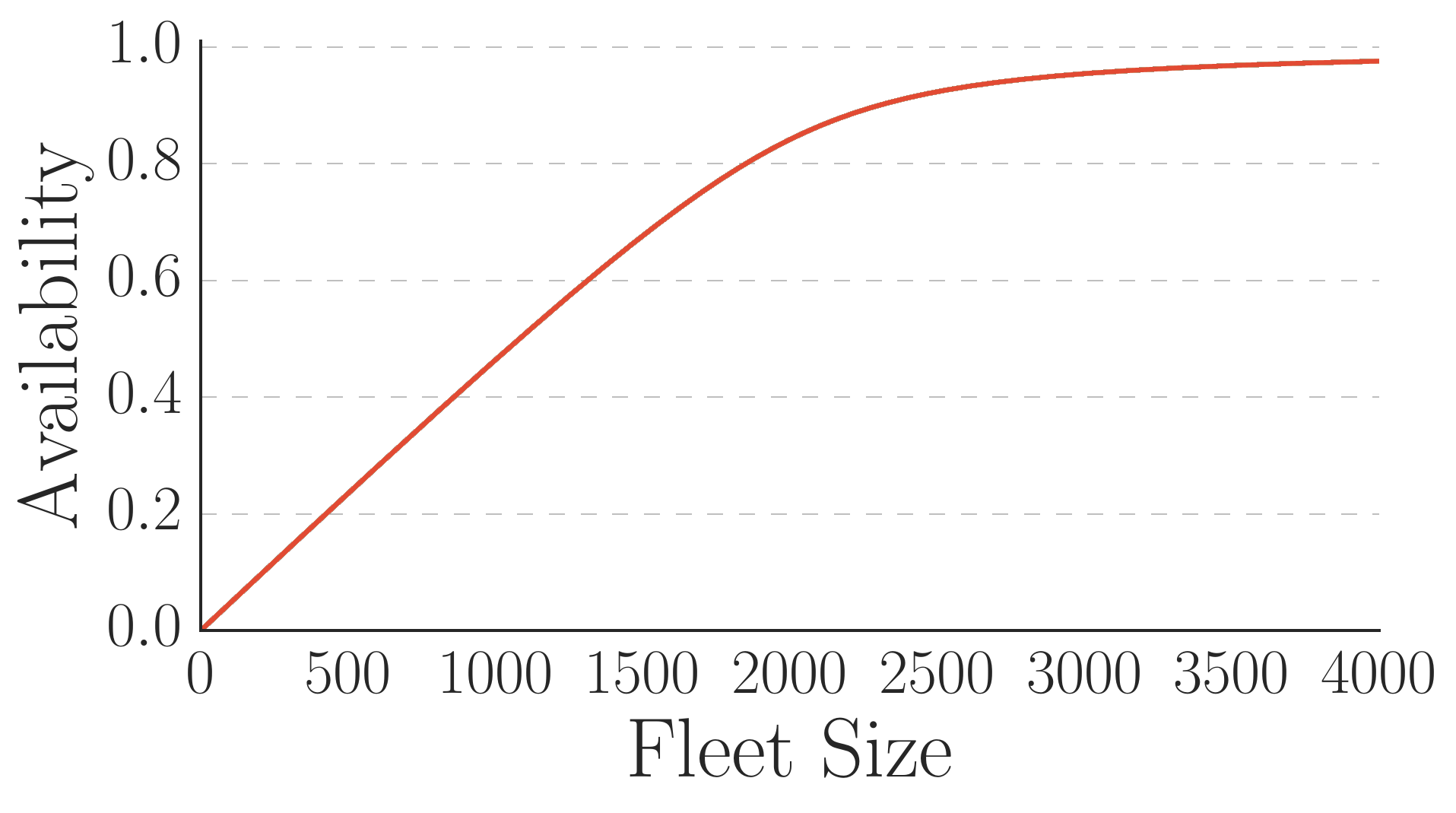}\label{fig:manhattan_av}}
		\sidesubfloat[]{\includegraphics[width=.45\textwidth]{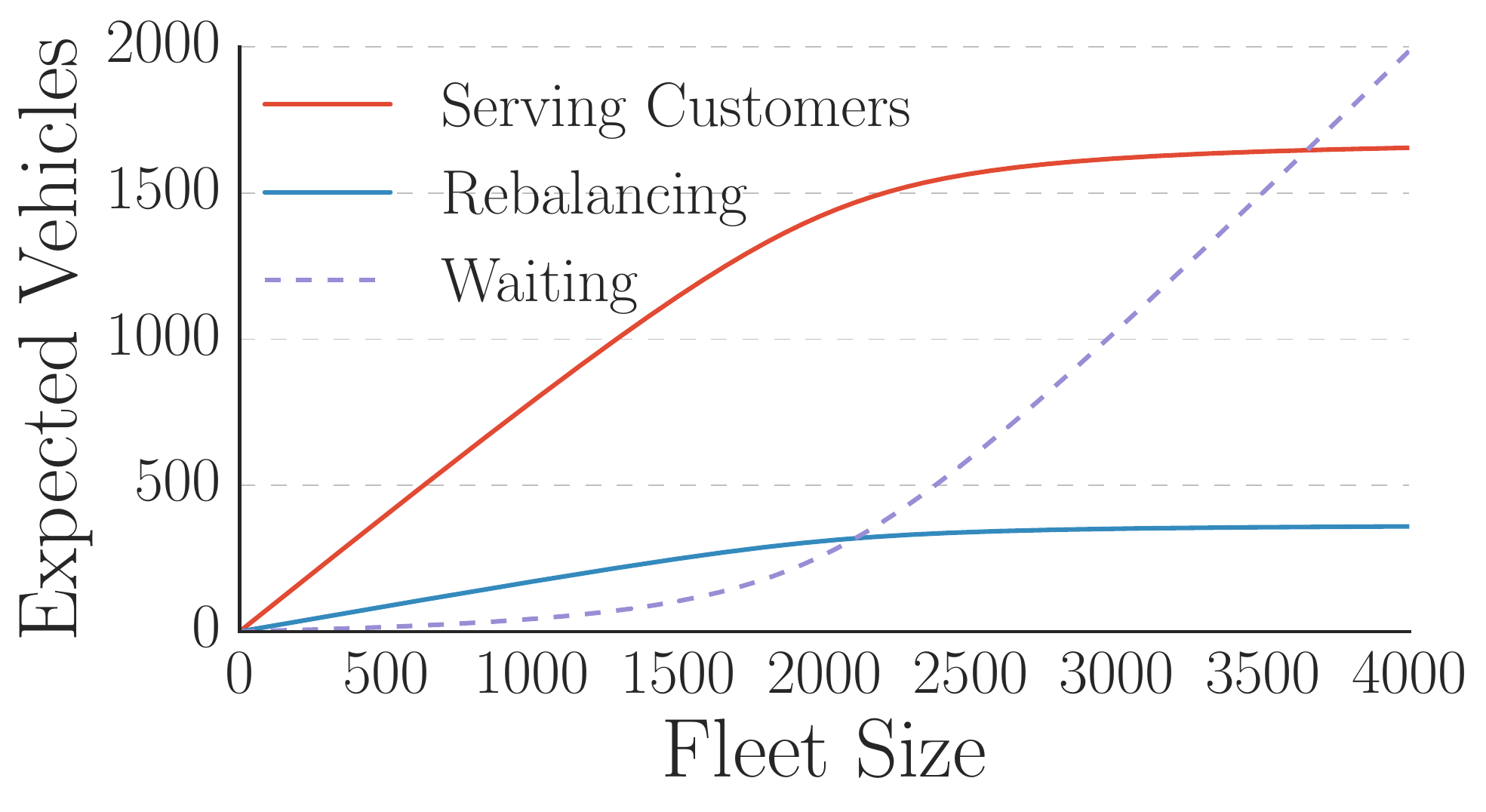}\label{fig:manhattan_veh}}\\
		\sidesubfloat[]{\includegraphics[width=.45\textwidth]{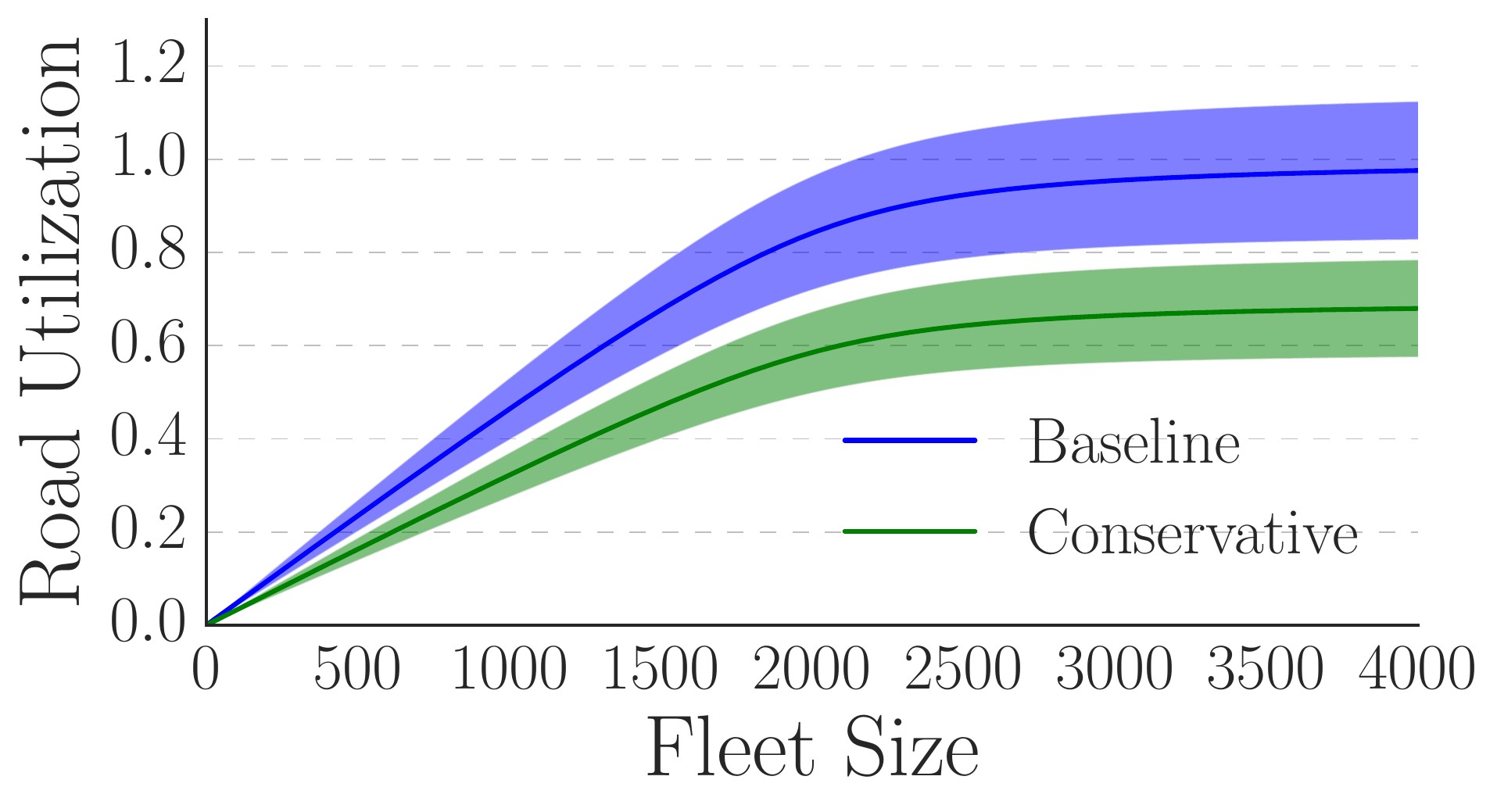}\label{fig:manhattan_ut}}
		\sidesubfloat[]{\includegraphics[width=.45\textwidth]{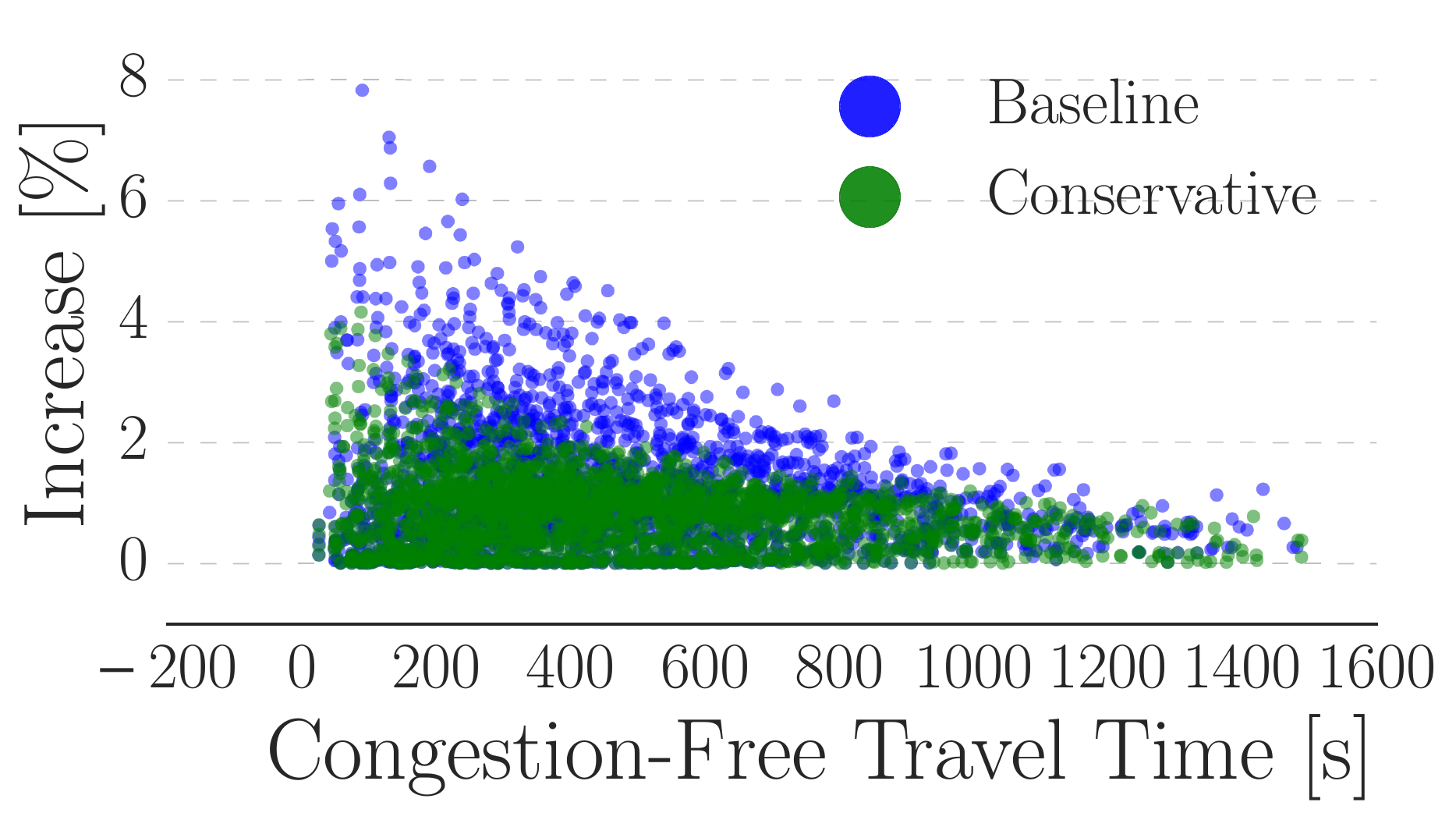}\label{fig:manhattan_traveltimes}}\\		
		\else
		\sidesubfloat[]{\includegraphics[width=.365\textwidth]{manhattan_availability.pdf}\label{fig:manhattan_av}}
		\sidesubfloat[]{\includegraphics[width=.365\textwidth]{manhattan_vehicles.pdf}\label{fig:manhattan_veh}}\\
		\sidesubfloat[]{\includegraphics[width=.365\textwidth]{manhattan_utilization_std.pdf}\label{fig:manhattan_ut}}
		\sidesubfloat[]{\includegraphics[width=.365\textwidth]{manhattan_traveltimes_cfr.pdf}\label{fig:manhattan_traveltimes}}\\
		\fi
		\caption{\protect\subref{fig:manhattan_av} Station availabilities as a function of fleet size for the baseline case.
			\protect\subref{fig:manhattan_veh} Expected number of vehicles by usage as a function of fleet size for the baseline case.
			\protect\subref{fig:manhattan_ut} Utilization as a function of fleet size for the most utilized road. The colored band denotes $\pm 1$ standard deviation from the mean.
			\protect\subref{fig:manhattan_traveltimes} Increase in expected travel time for each O-D pair when considering the BPR delay model.
			}
		\label{fig:sub1}
	\end{figure}
	
	\vspace{-10mm}
	\section{Conclusions}\label{conclusions}
	\vspace{-5mm}
	In this paper, we presented a novel queuing theoretic framework for modeling AMoD systems within capacitated road networks.
	We showed that, for the routing and rebalancing problem, the stochastic model we propose asymptotically recovers existing models based on the network flow approximation. 
	The model enables the analysis and control of the probabilistic distribution of the vehicles, and not only of its expectation: in particular (i) it enables the computation of higher moments of the vehicle distribution on road links and at stations and (ii) it allows to establish an arbitrary bound on the asymptotic probability of exceeding the capacity of individual road links.
	
	The flexibility of the model presented will be further exploited in future work. First, we would like to incorporate a more accurate congestion model, using load-dependent IS queues as roads, in order to study heavily congested scenarios. Second, we currently consider the system in isolation from other transportation modes, whereas, in reality, customer demand depends on the perceived quality of the different transportation alternatives. Future research will explore the effect of AMoD systems on customer behavior and how to optimally integrate fleets of self-driving vehicles with existing public transit. Third, we would like to examine scenarios where the vehicle fleet is electric-powered and explore the relationship between the constraints imposed by battery charging and the electric grid. Fourth, the current model assumes that each customer travels alone: future research will address the problem of \emph{ride-sharing}, where multiple customers may share the same vehicle. Lastly, the control policy proposed in this paper is open-loop, and thus sensitive to modeling errors (e.g., incorrect estimation of customer demand). Future research will characterize the stability, persistent feasibility and performance of \emph{closed-loop} model predictive control schemes based on a receding-horizon implementation of the controller presented in this paper.

	\begin{acknowledgement}
	The authors would like to thank the National Science Foundation for funding this work via the NSF CAREER award.
	\end{acknowledgement}
	\bibliographystyle{splncs03} 
	\bibliography{../../../bib/main_legacy.bib}

	\ifarxiv
	\section*{Appendix}
	\paragraph{Proof of Lemma \ref{le:folding}}
		Using the routing matrix specified in Equation \eqref{eq:routing} we can rewrite the class throughputs (\ref{eq:throughput}) as
		\begin{equation}\label{eq:simple_pi_k}
		\small
		\begin{split}
		\pi_{i,k} & = \sum_{k'= 1}^{K} \sum_{j = 1}^{N} \pi_{j,k'} p_{j,k';i,k} = \sum_{k'\in \mathcal{D}_i} \sum_{j \in N_{in}(j)} \pi_{j,k'} p_{j,k';i,k}, \\
		& = \sum_{k'\in \mathcal{D}_i} \sum_{j \in N_{in}(j)} \pi_{j,k'} \widetilde{p}_i^{(k)} = \widetilde{p}_i^{(k)} \sum_{k'\in \mathcal{D}_i} \sum_{j \in N_{in}(j)} \pi_{j,k'}, \\
		\end{split}
		\end{equation}
		where the second equality acknowledges the fact that only queues feeding into $i$ and vehicles whose class destination is $i$ will be routed to $i$, and the third and fourth equalities take advantage of the fact that the probability of switching into class $k$ at queue $i$ is the same regardless of the original class $k'$. This allows to rewrite the total relative throughput
		\begin{equation}\label{eq:simple_pi}
		\small
		\begin{split}
		\pi_i & = \sum_{k = 1}^{K} \widetilde{p}_i^{(k)} \sum_{k'\in \mathcal{D}_i} \sum_{j \in N_{in}(j)} \pi_{j,k'} = \sum_{k'\in \mathcal{D}_i} \sum_{j \in N_{in}(j)} \pi_{j,k'},
		\end{split}
		\end{equation}
		since $\sum_{k = 1}^{K} \widetilde{p}_i^{(k)} = 1$. As a consequence of (\ref{eq:simple_pi}) and (\ref{eq:simple_pi_k}), we can relate the class relative throughputs to the total relative throughputs
		
		\begin{equation}\label{eq:folded_traffic}
		\small
		\pi_{i,k} = \widetilde{p}_i^{(k)} \pi_i.
		\end{equation}
		
		Now, assume the relative throughputs belong to a feasible solution to OSCARR. We proceed to reduce (\ref{eq:routing_constraint}) by using the routing matrix:
		\begin{equation}
		\small
		\begin{split}
		\pi_{s^{(k)}, k} & = \sum_{k' \in \mathcal{K}} \sum_{j \in \mathcal{N}} \pi_{j,k'} p_{j,k;t^{(k)},k'} = \sum_{k' \in \mathcal{K}} \widetilde{p}_{t^{(k)}}^{(k')} \sum_{j \in \mathcal{N}_{in}(t^{(k)})} \pi_{j,k'} = \sum_{j \in \mathcal{N}_{in}(t^{(k)})} \pi_{j,k'},
		\end{split}
		\end{equation}
		by inserting this into (\ref{eq:simple_pi}) and applying (\ref{eq:folded_traffic}) we obtain
		\begin{equation}
		\small
		\begin{split}\label{eq:folded_traffic_origins}
		\pi_ i & = \sum_{k \in \mathcal{D}_i} \pi_{s^{(k)}, k} = \sum_{k \in \mathcal{D}_i} \widetilde{p}_{s^{(k)}}^{(k)} \pi_{s^{(k)}}.
		\end{split}
		\end{equation}
\paragraph{Proof of Lemma \ref{le:lambdas_constr}}

	The proof of Lemma \ref{le:lambdas_constr} is very similar to Theorem 4.3 in \cite{RZ-MP:15_MODa}.

		Consider the case where (\ref{eq:lambdas_constr}) holds. We can write (\ref{eq:folded_traffic_origins_lemma}) in terms of the relative utilization rate:
		\begin{equation}\label{eq:folded_traffic_origins_utilization}
		\small
		\left(\sum_{k \in \mathcal{D}_i} \lambda_{s^{(k)}}^{(k)}\right) \gamma_i = \sum_{k \in \mathcal{D}_i} \gamma_{s^{(k)}} \lambda_{s^{(k)}}^{(k)}\,.
		\end{equation}
		
		Now, by grouping customer and rebalancing classes by origin-destination pairs, we define $\varphi$ as
		\begin{equation}
		\small
		\varphi_{ij} = \lambda_j^{(q)} + \lambda_j^{(r)},
		\end{equation}
		such that $s^{(q)}=s^{(r)}=j$ and $t^{(q)}=t^{(r)}=i$. Additionally, let $\zeta_{ij} = \varphi_{ij}/\sum_j \varphi_{ij}$. We note that there are no classes for which $s^{(k)}=t^{(k)}$, so we set $\varphi_{ii}=\zeta_{ii}=0$. Under this definition, the variables $\{\zeta_{ij}\}_{ij}$ represent an irreducible Markov chain. Thus, Equation (\ref{eq:folded_traffic_origins_utilization}) can be rewritten as $\gamma_i = \sum_j \gamma_j \zeta_{ij}$ or more compactly $Z \gamma = \gamma$, where the rows of $Z$ are $[\zeta_{i1},\zeta_{i2},...,\zeta_{iS}]$, with $S = |\mathcal{S}|, $ $i = 1,..., S$, and $\gamma = (\gamma_1, ..., \gamma_s)$. This result is identical to \cite{RZ-MP:15_MODa}. Since $Z$ is an irreducible, row stochastic Markov chain, by the Perron-Frobenius theorem the unique solution is given by $\gamma = (1, ... ,1)^T$. Thus, $\gamma_i = \gamma_j$ for all $i$.
		
		On the other hand, we consider again Equation (\ref{eq:folded_traffic_origins_utilization}). If the network $\mathcal{Z}$ is a solution to problem (\ref{eq:cbrrp}), then for all $i,j \,\,\gamma_i = \gamma_j = \gamma$, and (\ref{eq:folded_traffic_origins_utilization}) becomes
		
		\begin{equation}
		\small
		\gamma \widetilde{\lambda}_i = \gamma \sum_{k \in \mathcal{D}_i} \lambda_{s^{(k)}}^{(k)}, \quad
		\widetilde{\lambda}_i = \sum_{k \in \mathcal{D}_i} \lambda_{s^{(k)}}^{(k)}.
		\end{equation}

	\paragraph{Bureau of Public Roads delay model} The Bureau of Public Roads (BPR) delay model is a commonly used equation for relating traffic to travel time \cite{BPR:64}. Under this model, the travel time on a road link is given by
	\begin{equation}
	T_i' = T_i \left(1 + \delta \left(\frac{x_i}{C_i}\right)^\beta\right),
	\end{equation}
	where $T_i'$ is the real mean travel time, $T_i$ is the free flow travel time, $x_i$is the number of vehicles on the road, $C_i$ the nominal capacity of the road, and $\delta$ and $\beta$ are function parameters usually set to 0.15 and 3, respectively.
	\fi
\end{document}